\documentclass [12pt]{article}
\usepackage{amssymb,amsmath,amsthm}
\usepackage{graphicx}

\usepackage[cp1250]{inputenc}
\usepackage[active]{srcltx}

\long\def\comment#1{{}}

\newtheorem{thm}{\bf Theorem}

\begin{document}

\title{{\bf On Axiomatization of Inconsistency Indicators for Pairwise Comparisons}}

\author{W.W. Koczkodaj
\thanks{
Computer Science, Laurentian University,
Sudbury, Ontario P3E 2C6, Canada 
wkoczkodaj@cs.laurentian.ca 
} \\
\and R. Szwarc
\thanks{Institute of Mathematics, University of Wroclaw,
Wroclaw, Poland, szwarc2@gmail.com}}

\maketitle

\begin{abstract}

This study examines the notion of inconsistency in pairwise comparisons for providing an axiomatization for it. 
It also proposes two inconsistency indicators for pairwise comparisons. 
The primary motivation for the inconsistency reduction 
is expressed by a computer industry concept ``garbage in, garbage out''.  
The quality of the output depends on the quality of the input.  \\


\end{abstract}

\section{Introduction}

The method of pairwise comparisons (PC method here) is attributed to Fechner (see \cite{Fech1860}) as a formal scientific method although it was first mentioned by Condorcet in \cite{Condorcet1785} who only used it in its primitive form: win/loss. However, Thurstone (see \cite{Thur27}) proposed what is known as ``The Law of Comparative Judgments'' in 1927. In 1977, Saaty proposed what is known as the Analytic Hierarchy Process (AHP) method based on modified pairwise comparisons with a hierarchy structure in \cite{Saaty77}. In this study, however, the hierarchy is not considered. 

Saaty's study \cite{Saaty77} had a profound impact on the pairwise comparisons research. However, his AHP should not be equalized with pairwise comparisons, despite using them. The restrictions assumed by Saaty (e.g., fixed scale: 1 to 9) probably serves its proponent well for whatever purpose he has designed it. AHP is a subset of the pairwise comparisons method which does not assume any particular scale. A proof was provided in \cite{FKS2010} that a small scale (1 to 3) has desired  mathematical properties for the use in pairwise comparisons.

It is also worth to note that  this study considers only the multiplicative PC which is based on \textit{\textit{``how many times?''}}, while the additive version of pairwise comparisons (``by how much...'')was recently analyzed in \cite{KKFY2012}. It has a different type of inconsistency (not addressed here). 

Recently, the study \cite{K2013a} presents an innovate iterative heuristic rating estimation algorithm that tries to deal with the situation when exact estimations for some concepts (stimulus) $C_{K}$ are a priori known and fixed, whilst the estimates for the others (unknown concepts $C_{U}$) need to be computed. The relationship between the local estimation error, understood
as the average absolute error $E(c)$ over all direct estimates for
the concept $c\in C_{U}$ and the pairwise comparisons matrix inconsistency index is shown. 

Regretfully, pairwise comparisons theory is not as popular as in mathematics, for example, partial differential equations, hence basic concepts need to be presented in the next section but it is not PC method experts.

\section{Pairwise comparisons basics}
\label{sec:bc}

An $N \times N$ pairwise comparison matrix simply is a square matrix $M=[m_{ij}]$ such that $m_{ij}>0$ for every 
$i,j=1, \ldots ,n$. A pairwise comparison matrix $M$ is called {\em 
reciprocal} if $m_{ij} = \frac{1}{m_{ji}}$ for every $i,j=1, \ldots ,n$
(then automatically $m_{ii}=1$ for every $i=1, \ldots ,n $). Let
us assume that:

\begin{displaymath}
M = \begin{bmatrix}
1 & m_{12} & \cdots & m_{1n} \\ 
\frac{1}{m_{12}} & 1 & \cdots & m_{2n} \\ 
\vdots & \vdots & \vdots & \vdots \\ 
\frac{1}{m_{1n}} & \frac{1}{m_{2n}} & \cdots & 1
\end{bmatrix}
\end{displaymath}

\noindent where $m_{ij}$ expresses a relative preference of entity (or stimuli) $s_i$ over $s_j$.

A pairwise comparison matrix $M$ is called consistent (or transitive) if 
$$m_{ij} * m_{jk}=m_{ik}$$ for every $i,j,k=1,2, \ldots ,n$. \\

We will refer to it as a ``consistency condition''.
While every consistent matrix is reciprocal, the converse is false in general. If the consistency condition does not hold, the matrix is inconsistent (or intransitive).
   
Consistent matrices correspond to the ideal situation in which there are the exact values $s_1, \ldots , s_n$ for the stimuli. 
The quotients $m_{ij}=s_i/s_j$ then form a consistent matrix. 
The vector $s=[s_1, \ldots s_n]$ is unique up to a multiplicative constant. 
The challenge of the pairwise comparisons method comes from the lack of consistency of the pairwise comparisons matrices which arise in practice (while as a rule, all the pairwise comparisons matrices are reciprocal). Given an $n \times n$ matrix $M$, which is not consistent, the theory attempts to provide a consistent $n \times n$ matrix $M'$ which differs from matrix $M$ ``as little as possible''. 

 The matrix: $M= s_i/s_j$ is consistent for all (even random) values $v_i$. It is an important observation since it implies that a problem of approximation is really a problem of a norm selection and the distance minimization. For the Euclidean norm, the vector of geometric means (equal to the principal eigenvector for the transitive matrix) is the one which generates it. Needless to say that only optimization methods can approximate the given matrix for the assumed norm (e.g., LSM for the Euclidean distance, as recently proposed in \cite{AZG2012}). Such type of matrix is examined in \cite{JT2011} as ``error-free'' matrix.

It is unfortunate that the singular form ``comparison'' is sometimes used considering that a minimum of three comparisons are needed for the method to have a practical meaning. Comparing two entities (stimuli or properties) in pairs is irreducible, since having one entity compared with itself gives trivially 1. Comparing only two entities ($2 \times 2$ PC matrix) does not involve inconsistency. Entities and/or their properties are often called stimuli in the PC research but are rarely used in applications.

\section{The pairwise comparisons inconsistency notion}
\label{sec:ic}

The study \cite{Saaty77} includes: ``We may assume that when the inconsistency indicator shows the perturbations from consistency are large and hence the result is unreliable, the information available cannot be used to derive a reliable answer.'' 

The above quotation is consistent with the popular computer adage GIGO (garbage in -- garbage out). GIGO summarizes what has been known for a long time: getting good results from ``dirty data'' is unrealistic, and surely, cannot be guaranteed. An approximation of a pairwise comparisons matrix is meaningful if the inconsistency is acceptable. It can be done by localizing the inconsistency and reducing it to a certain predefined threshold. For the time being, the inconsistency threshold is arbitrary or set by a heuristic, since there is no theory to find it. It is a similar situation to p-value in statistics -- often assumed as 0.05 (or any other arbitrary value), but can be undermined for each individual case. 

As pointed out earlier, given an inconsistent matrix $M$, the theory attempts to approximate it with a consistent matrix $M'$ that differs from matrix $M$ ``as little as possible''. 
The consistency of a matrix A, expressed by $m_{ij}*m_{jk}=m_{ik}$, was called in \cite{Saaty77} a ``cardinal consistency''. In this study, a term ``triad'' is used for $(m_{ij},m_{ik},m_{jk})$ (these three matrix elements in the above cardinal consistency condition).

Before progressing to a formal inconsistency definition, the most important question needs to be addressed: \textit{``where does the inconsistency come from?''}
The short answer to this question is from the excess of input data. The superfluous data comes from collecting data for all pairs combinations which is $n*(n-1)/2$, while only $n-1$ proper comparisons (e.g., the first row or column and even diagonals or some of their combinations) would suffice. The inconsistency in a triad is illustrated by the following example.

\noindent \textbf{Example:}\\
This is an inconsistent matrix $M$, $3 \times 3$ with one triad $(2,2,2)$, which is marked by the bold font, is:

\begin{displaymath}
  A = \begin{bmatrix}
    1    & \textbf{2}    & \textbf{2} \\
    1/2  & 1    & \textbf{2} \\
    1/2  & 1/2  & 1 \\
  \end{bmatrix}
\end{displaymath}

\noindent Evidently, matrix $A$ displays an abnormality since $2*2 \ne 2$.
The computed vector of weight ($s_i$ mentioned earlier in this section)is:

$$s=[0.4934, 0.3108, 0.1958]$$

\noindent The above values generate the fully consistent PC matrix B:

\begin{displaymath}
  B = \begin{bmatrix}
    1    & \textbf{1.5874011}    & \textbf{2.5198421} \\
    0.6299605  & 1    & \textbf{1.5874011} \\
    0.3968503  & 0.6299605  & 1 \\
  \end{bmatrix}
\end{displaymath}

Everything comes back to normality when  $a_{1,3}$ is changed from 2 to 4. Although this is a rather simple example, the proposed inconsistency reduction process comes to finding such a triad and changing an offending value with the value which making the consistency condition to hold or at least to have one side of the consistency condition close to the other side.

Table~\ref{fig:T7a} shows three triads consisting of matrix elements, which may not be neighbors in this matrix. Different types of parenthesis have been used for each triad, only for easier demonstration. All triads above the main diagonal have the carpenter angle tool shape or the mirror image of the capital letter ``L'', with the middle value in the ``elbow'' element ideally (for the consistency) being the product of the outer elements.

\begin{table}[h]
\centering
    \begin{tabular}{|c|c|c|c|c|c|c|}
        \hline
        ~~1~~ & ~~~~~ & (1,3) & ~~~~~ & ~~~~~ & ~~~~~ & (1,7) \\ \hline
        ~ & 1 & ~ & [2,4] & ~ & [2,6] & ~ \\ \hline
        ~ & ~ & 1 & ~ & ~ & ~ & (3,7) \\ \hline
        ~ & ~ & ~ & 1 & \{4,5\} & [4,6] & \{4,7\} \\ \hline
        ~ & ~ & ~ & ~ & 1 & ~ & \{5,7\} \\ \hline
        ~ & ~ & ~ & ~ & ~ & 1 & ~ \\ \hline
        ~ & ~ & ~ & ~ & ~ & ~ & 1 \\
        \hline
    \end{tabular}
\caption{PC matrix with various triads}
\label{fig:T7a}
\end{table}

Triads may have one overlapping matrix element. For  example, $i=1$, $j=2$, and $k=3$ creates a triad with one element in the triad created by $i=1$, $j=3$, and $k=7$. According to the triad production expression:  $(a_{ij},a_{ik},a_{jk})$,  it is element $a_{1,3}$. Evidently, triad elements do not need to be neighbors in the matrix, but if they are, they must be just above the main diagonal, as illustrated by Table~\ref{fig:T7b}. 

\begin{table}[h]
\centering
    \begin{tabular}{|c|c|c|c|c|c|c|}
        \hline
        ~~1~~ & (1,2) & (1,3) & ~ & ~ & ~ & ~ \\ \hline
        ~     & 1 & (2,3) & (2,4) & ~ & ~ & ~ \\ \hline
        ~     & ~ & 1 & (3,4) & (3,5) & ~ & ~ \\ \hline
        ~     & ~ & ~ & 1 & (4,5) & (4,6) & ~ \\ \hline
        ~     & ~ & ~ & ~ & 1 & (5,6) & (5,7) \\ \hline
        ~     & ~ & ~ & ~ & ~ & 1 & (6,7) \\ \hline
        ~     & ~ & ~ & ~ & ~ & ~ & 1 \\
        \hline
    \end{tabular}
\caption{All triads in a $7 \times 7$ matrix with elements which are neighbors}
\label{fig:T7b}
\end{table}

Inconsistent assessments cannot be accurate but after approximation, they may be closer to real values. 
Let us assume that the triad $(2,5,3)$ in Fig.~\ref{fig:triad253} reflects comparisons of three bars with lengths: A, B, and C made by experts on three different continents by the Internet.
Expert 1 compares A to B giving $A/B=3$ and Expert 2 compares B to C giving $B/C=2$.
One could object to $A/C=5$ given by Expert 3 after A to C are compared. Evidently, $A/B * B/C$ is  $A/C$, hence the result is $2*3=6$. However, we really do not know and will never know who made an estimation error! In fact, we can safely assume that each expert made ``just a little bit of error''. In particular, none of these three values could be accurate. It cannot be solved by any theory. A solution is needs to be found on individual basis for each application.

\begin{figure}[htb]
\centering
\includegraphics[scale=0.5]{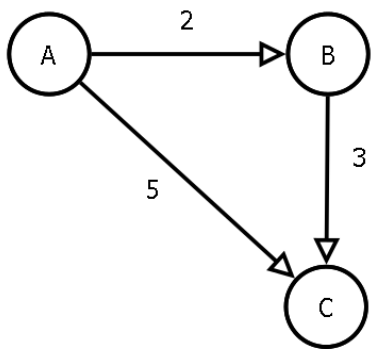}
\caption{A graphical representation of the triad (2,5,3)}
\label{fig:triad253}
\end{figure}

In this study, the approximation error (the most common in science and engineering) will be used and presented as a percentage. It will be simply called ``the error''. 
The approximation error in inaccurate data is the discrepancy between an exact value and some approximation to it.

Given some value $v$ and its approximation $v_{approx}$, the absolute error is: $\Delta = | v - v_{approx}|$
where the vertical bars denote the absolute value. 
For $v \neq 0$, the approximation error is defined as:

$$\delta = { | v - v_{approx}| \over | v | } =
 \left|  { v - v_{approx} \over v } \right|  = 
 \left|   1 - {v_{approx} \over v } \right| 
$$

Each triad generates a PC matrix $M$ of the size $3 \times 3$. Let us use A, B, and C to reflect lengths of three bars. The value $M[1,2]=1$ represents $A = B$, $M[2,3]=1$ represents $B = C$ hence the expectation is $A=C$ but the third estimates is 5. It is reflected by the last bar hence the error is 500\%. As assumed, $x$ can take any arbitrary value and so can the estimation error.  
For small values of $n$, the maximum value of the error, still acceptable by the eigenvalue-based inconsistency, has been presented in Tab.~\ref{tab:err}. PC matrix with triads $(1,x,1)$ is of a considerable importance and it is analyzed in Section~\ref{sec:cpc}.

\section{Axiomatization of inconsistency}

It is generally assumed that it was Saaty who in \cite{Saaty77} defined PC matrix $A$ as consistent if and only if 
$a_{ij} *  a_{jk} = a_{ik}$ for $i,j,k = 1, 2,..., n$. 
However, inconsistency was defined and examined before 1977, by at least these four studies published between 1939 and 1961: \cite{KSB1939, H1953, GS1958, S1961}. To our knowledge, no axiomatization has ever been proposed for the general case of pairwise comparisons matrix with real positive entries, although it seems that attempts have been made for matrices with integer values  for win-tie-loss entries. \\

The common sense expectations for the inconsistency indicator $ii$ of a triad $T=(x,y,z)$ are:
\begin{enumerate}
    \item $ii=0$ for $y=x*z$,
	 \item $ii \in [0,1)$ - by common sense, wan ``ideal inconsistency'' cannot be  achieve,
    \item for a consistent triad $ii(x,y,z)=0$ with $xz=y$, increasing or decreasing $x,y,z$ results in increasing $ii(x,y,z)$.

\end{enumerate}

The third axiom is crucial for any axiomatization.
Without this axiom, an inconsistency indicator would not make practical sense. For any assumed definition for inconsistency, an inconsistency indicator of a triad $T'=(x',y',z')$ cannot be smaller than of $T=(x,y,z)$ if it is worse by one of more coordinates, which is what the third axiom is about. That is, $ii(x',y',z') \ge ii(x,y,z)$.
It is a reasonable expectation that the worsening of a triad, used in the definition of consistency (also in \cite{Saaty77}), cannot make the entire matrix more consistent.

\noindent For $ii(x,y,z) > 0$, there are two cases:

\begin{description}
    \item[(a)] $xz<y$
    \item[(b)] $xz>y$
\end{description}

\noindent In case of:
\begin{description}
    \item[(a)] if $x'z'< xz \& y'>y$ then $ii(x,y,z)< ii(x',y',z')$ 
    \item[(b)] if $x'z'> xz \& y'<y$ then $ii(x,y,z)< ii(x',y',z')$ 
\end{description}

\noindent Let us look at the following two examples:

\begin{itemize}
\item 
$ii(1.5,2,2.5)$ will increase if 1.5 or 2.5 are increased, since 1.5*2.5 is already greater than 2. On the other hand, decreasing 2 should also increase the inconsistency.
\item $ii(1.5,2.5,1.2)$ will increase if  2.5 in increased, since it is greater than 1.5*1.2=1.8, but decreasing 1.5 or 1.2 should also increase inconsistency for the same reason.
\end{itemize}

\noindent Based on the proposed axioms for inconsistency and \cite{Kocz93}, let us define:

$$f(x,y,z) = 1-\min\left \{\frac{y}{xz},\frac{xz}{y}\right \}$$. 

\noindent It is equivalent to:

$$f(x,y,z)=1- e^{-\left|\ln\left (\frac{y}{xz}\right )\right |}$$.

\noindent The expression $|\ln(\frac{y}{xz})|$
is the distance of the triad $T$ from 0. When this distance increases, the $f(x,y,z)$ also increases. It is important to notice here that this definition allows us to localize the inconsistency in the matrix PC and it is of a considerable importance for most applications.  

Another possible definition of the inconsistency has a global character and needs a bit more explanations.  Let $A=\{a_{ij}\}_{i,j=1}^n$ be a reciprocal positive matrix. The matrix $A$ is consistent if and only if for any $1\le i<j\le n$ the following equation holds:

$$a_{ij}=a_{i,i+1}a_{i+1,i+2}\ldots a_{j-1,j}.$$
Therefore, the inconsistency indicator of $A$ can be also defined as:
$${\rm ii}(A)=1-\min_{1\le i<j\le n}  \min\left ({a_{ij}\over a_{i,i+1}a_{i+1,i+2}\ldots a_{j-1,j}},\,
 {a_{i,i+1}a_{i+1,i+2}\ldots a_{j-1,j}\over a_{ij}} \right ) $$

\noindent It is equivalent to:

$${\rm ii}(A)=1-\max_{1\le i<j\le n} 
   \left (1 - e^{-\left |\ln \left ( {a_{ij}\over 
   a_{i,i+1}a_{i+1,i+2}\ldots  a_{j-1,j}}  \right )\right |}\right ) 
$$

Both $ii$ definitions have some advantages and disadvantages. The first definition allows us to find the localization of the inconsistency. The second definition may be useful when the global inconsistency is more important. The first definition follows what is adequately described by the idiom: ``one bad apple spoils the barrel''. A hybrid of using two definitions may be a practical solution in applications. Alternatively, both definitions can be used in a sequence.

\section{The analysis of $CPC(x,n)$ matrix}
\label{sec:cpc}

In this section, a pairwise matrix with all 1s except for two corners (called ``corner comparisons matrix or CPC'') is analyzed. 
Consider the matrix $CPC(x,n),$ with $x>1,$ defined by
\begin{displaymath}
CPC(x,n)=\begin{bmatrix}1&1&\cdots&1&x\\
1&1&\cdots&1&1\\
\vdots &\vdots &\ddots&\vdots &\vdots\\
1&1&\cdots&1&1\\
x^{-1}&1&\cdots&1&1\\\end{bmatrix}\in M_{n\times n}(\mathbb R)
\end{displaymath}
By the Perron-Frobenius theorem, the principal eigenvalue $\lambda_{\rm max}$ corresponds to a unique (up to constant multiple) eigenvector $w=\{w_i\}_{i=1}^n$ with positive entries. Since the rows $r_2,r_3\ldots, r_{n-1}$ of the matrix $CPC(x,n)$ are equal the eigenvector, $w$ satisfies $w_2=w_3=\ldots =w_{n-1}.$ After normalization, it may be assumed that
$$ w=(a,1,1,\ldots, 1, b).$$
The eigenvalue equation $CPC(x,n) w=\lambda_{\rm max}  w$ is reduced to the system of three equations with three unknown $a,\,b$ and $\lambda_{\rm max}.$
\begin{eqnarray*}
a+ n-2 + bx&=&\lambda_{\rm max} a,\\
a+n-2 + b&=&\lambda_{\rm max},\\
{a\over x} + n-2 + b&=&\lambda_{\rm max} b.
\end{eqnarray*}
By solving the system consisting of the first and the last linear equations, relative to $a$ and $b,$ we get
$$a=(n-2)\,{{x+\lambda_{\rm max} -1\over \lambda_{\rm max}^2-2\lambda_{\rm max}}},\quad  
b=(n-2)\,{  x^{-1}+\lambda_{\rm max} -1\over \lambda_{\rm max}^2-2\lambda_{\rm max}}.$$
Substituting $a$ and $b$ in the second equation by the above expressions (after some transformations), the following third degree equation for $\lambda_{\rm max}$ is obtained:
\begin{equation}
\lambda_{\rm max}^3-n\lambda_{\rm max}^2= (n-2)(x^{-1}+x-2).
\end{equation}
It can still be transformed that into
$$
{\lambda_{\rm max}-n\over n-1}={n-2\over n-1}\,{x^{-1}+x-2\over \lambda_{\rm max}^2}.
$$
Since the right hand side is positive, we must have $\lambda_{\rm max}>n.$ \\
Therefore
\begin{equation}\label{ineq1}
{\lambda_{\rm max}-n\over n-1}\le {n-2\over n-1}\,{x^{-1}+x-2\over n^2}.
\end{equation}

\noindent It has been assumed that $x>1$ therefore 
$x^{-1} <1$ 

\noindent also

$${n-2\over n-1} < 1$$

\noindent hence the following inequality holds:

\begin{equation}\label{ineq}
{\lambda_{\rm max}-n\over n-1}\le {x\over n^2}.
\end{equation}

The inequality (3) has a very important implication. No matter how large $x$ is, there is always such $n$ that the left hand side of (3) is as small as it can be assumed. So, regardless of the assumed threshold in \cite{Saaty77} (de facto, originally set to 10\%), the matrix is acceptable according to the consistency rule set in \cite{Saaty77}.

Evidently, the arbitrarily large $x$ in the matrix $CPC(x,n)$ of size $n \times n$ invalidates the acceptability of this matrix. Hence, by a \textit{reductio ad absurdum}, the soundness of the eigenvalue-based inconsistency indicator represented by the left hand side inequality (3) must be 
dismissed. \\

\noindent {\bf Example:}\\

For $n=6$ and $x=6$:
$${\lambda_{\rm max}-n\over n-1} \le {4\over 5}\, {4+(1/6)\over 36}=0.0925925...$$
Actually, we can determine numerically that $\lambda_{\rm max}= 6.406123...$ \\
Then
$$ {\lambda_{\rm max}-n\over n-1}=0.081224...
$$

Now, general reciprocal matrices will be considered. By a careful analysis of \cite{Saaty77}, the following lower estimates for
$\lambda_{\rm max}$ for general reciprocal positive matrices are obtained:
\begin{thm}
Let $A=\{a_{ij}\}_{i,j}^n$ be a
reciprocal matrix with positive entries.
Then
$$\lambda_{\rm max}\ge n +  {1\over 3n}\,{ {\rm ii}^2(A)\over{\sqrt[3]{ 1-{\rm ii}(A) }}} ,$$
where
$${\rm ii}(A) =1 -\min_{i<k<j}\min \left \{ {a_{ij}\over a_{ik}a_{kj}},{a_{ik}a_{kj}\over a_{ij}}\right \} .$$
\end{thm}
\begin{proof}
Let $w=\{w_i\}_{i=1}^n$ be the eigenvector corresponding to the eigenvalue $\lambda_{\rm max}.$
By the Perron-Frobenius theory, we have $w_i>0.$    
Thus
$$\lambda_{\rm max}\, w_i= \sum_{j=1}^n a_{ij}w_j.$$
By an easy transformation and the fact that $a_{ii}=1$ (see  \cite{Saaty77}, pages 237-238), we get
$$
n\lambda_{\rm max}-n=\sum_{1\le i<j\le n} \left (a_{ij}{w_j\over w_i} + a_{ji}{w_i\over w_j}\right ).$$
This implies 
\begin{equation}\label{sum}
n(\lambda_{\rm max}-n) = \sum_{1\le i<j\le n} \left (a_{ij}{w_j\over w_i} + a _{ji}{w_i\over w_j}-2\right )
\end{equation}
Let us assume that the maximal inconsistency is attained at the triad $s<u<t,$ i.e.
$$ {\rm ii}(A)=1-\min \left \{ {a_{st}\over a_{su}a_{ut}},{a_{su}a_{ut}\over a_{st}}\right \}.$$
Every term in the sum of (\ref{sum}) is nonnegative as $x+x^{-1}-2\ge 0, $ for $x>0$ and $a_{ji}=a_{ij}^{-1}.$
By reducing the sum to three terms corresponding to the triad $s<u<t$, we get
\begin{equation}\label{three}
n(\lambda_{\rm max}-n)\ge a_{su}{w_u\over w_s} + a_{us}{w_s\over w_u}+a_{ut}{w_t\over w_u} + a_{tu}{w_u\over w_t}
+a_{st}{w_t\over w_s} + a_{ts}{w_s\over w_t}-6. 
\end{equation} 
Denote
$$ x=a_{su}{w_u\over w_s},\ y= a_{ut}{w_t\over w_u}, \ \alpha= {a_{su}a_{ut}\over a_{st}}.$$ Then
the right hand side of (\ref{three}) is given by
$$ f(x,y):=x+x^{-1}+y+y^{-1} +\alpha^{-1}xy + \alpha x^{-1}y^{-1} -6.$$
By calculating the partial derivatives of $f(x,y)$ and equating them to zero, we can easily determine
that the minimal value of $f(x,y)$ is attained for
$$x=y=\alpha^{1/3}.$$ 
We will consider the case $\alpha\le 1,$ i.e. ${\rm ii}(A)=1-\alpha$ (the other case $\alpha>1$ can be dealt with similarly). We have
\begin{multline*}f(x,y)\ge 3(\alpha^{1/3}+\alpha^{-1/3}) -6 =3\alpha^{-1/3}(1-\alpha^{1/3})^2
 \\
=3\alpha^{-1/3}\left ({1-\alpha\over 1+\alpha^{1/3}+\alpha^{2/3}}\right )^2\ge {1\over 3}\alpha^{-1/3}(1-\alpha)^2={1\over 3}\, { {\rm ii}^2(A)\over{\sqrt[3]{ 1-{\rm ii}(A) }}}.
\end{multline*}
Summarizing, we get
$$n(\lambda_{\rm max}-n)\ge  {1\over 3}\,{ {\rm ii}^2(A)\over{\sqrt[3]{ 1-{\rm ii}(A) }}},
$$ which yields the conclusion.
 \end{proof}
 \noindent{\bf Remark.} Theorem 1 yields
 $${\lambda_{\rm max}-n\over n-1}\ge  {1\over 3(n-1) n}\,{ {\rm ii}^2(A)\over{\sqrt[3]{ 1-{\rm ii}(A) }}}.$$
 Thus for given $n$ (say $n=6$), the quantity explodes if the indicator ${\rm ii}(A)$ approaches the value 1.

Another lower estimate for $\lambda_{\rm max}$ can be obtained. It takes into account the total inconsistency information of the matrix $A.$
\begin{thm} Let $T$ denote the set of all triads in the matrix $A$ and ${\rm ii}(t)$ be the inconsistency indicator of the triad $t,$
i.e. for $t=(i,k,j)$ with $i<k<j,$ let
$${\rm ii}(t)= 1 -\min\left \{ {a_{ij}\over a_{ik}a_{kj}},{a_{ik}a_{kj}\over a_{ij}}\right \}.$$
Then
$$ \lambda_{\rm max}\ge n + {1\over 3n(n-2)}\sum_{t\in T}{ {\rm ii}^2(t)\over{\sqrt[3]{ 1-{\rm ii}(t) }}}  .$$ 
\end{thm}
\begin{proof}
Every term $a_{uv}$ with $1\le u<v\le n$ belongs to $n-2$ triads. Therefore the formula (\ref{sum}) implies
\begin{multline*}
(n-2)n\,(\lambda_{\rm max}-n)\\=\sum_{i<k<j}\left [ a_{ik}{w_k\over w_i} + a_{ki}{w_i\over w_k}+a_{kj}{w_j\over w_k} + a_{jk}{w_k\over w_j}
+a_{ij}{w_j\over w_i} + a_{ji}{w_i\over w_j}-6\right ].
\end{multline*}
By the proof of Theorem 1, we get that for $\alpha=\min\{a_{ik}a_{kj}/a_{ij}, \,a_{ij}/a_{ik}a_{kj}  \}$ and $t=(i,k,j)$ we have
\begin{multline*}a_{ik}{w_k\over w_i} + a_{ki}{w_i\over w_k}+a_{kj}{w_j\over w_k} + a_{jk}{w_k\over w_j}
+a_{ij}{w_j\over w_i} + a_{ji}{w_i\over w_j}-6\\ \ge {1\over 3}\alpha^{-1/3}(1-\alpha)^2={1\over 3}\,{ {\rm ii}^2(t)\over{\sqrt[3]{ 1-{\rm ii}(t) }}}.
\end{multline*}
Hence
$$ (n-2)n\,(\lambda_{\rm max}-n)\ge {1\over 3}\sum_{t\in T} { {\rm ii}^2(t)\over{\sqrt[3]{ 1-{\rm ii}(t) }}}.$$
\end{proof}

The $CPC(x,n)$ matrix in the above example shows that for the eigenvalue-based consistency index (CI) an error of an arbitrary value is acceptable for the large enough $n$ (the matrix size). According to AHP theory, the $CPC(x,n)$ matrix is considered ``consistent enough'' (or ``good enough'') for $CI \leq 0.1$, although it has $n$ arbitrarily erroneous elements in it. The number $n$ of the erroneous elements grow to infinity with the growing $n$ and it invalidates using $CI$ for measuring the inconsistency.

\subsection{The interpretation of the $CPC(x,n)$ analysis}

Matrix $CPC(x,n)$ of the size of $3 \times  3$ has only one triad: $(1,x,1)$. Trivially, the only value of $x$ for this matrix to be consistent is 1 ($x=1*1$). For $x=2.62$, we have:

\begin{displaymath}
CPC(2.62,3) = \begin{bmatrix}
 1 & 1 & 2.62 \\
 1 & 1 & 1 \\
0.381679389 & 1 & 1 \\
\end{bmatrix}
\end{displaymath}

The principal eigenvalue of $CPC(2.62,3)$ is 3.10397 hence $CI=0.051985$ and it is less than 10\% of $RI=0.52$, hence acceptable due to the fact that the proposed consistency index (CI) is defined in \cite{Saaty77} as:

\begin{displaymath}
		CI = \frac{\lambda_{max} - n}{n-1}
\end{displaymath}

and the consistency ratio (CR) defined as 

\begin{displaymath}
		CR = \frac{CI}{RI}
\end{displaymath}

\noindent where $RI$ is the average value of $CI$ for random matrices and computed as 0.52 (decreased from 0.58 as stipulated in \cite{Saaty77}).

As previously observed, $x$ should be 1, so $x=2.62$ gives us 262\% error and it is still acceptable for the eigenvalue-based inconsistency. For matrices $3 \times 3$, $RI$ has been computed as 0.5245 hence $CR < 0.1$ for $CPC(2.62,3)$. The acceptable errors for other $n$ from 3 to 7 have been computed and presented in Tab.~\ref{tab:err}\\

\begin{table}[h]
  \centering
  \caption{Maximal errors acceptable by the eigenvalue-based inconsistency for $CPC(x,n)$} 
    \begin{tabular}{cc}
    \hline\hline
    n     & error for (1,x,1) \\
    \hline
    3     & 262\%   \\
    4     & 417\%   \\
    5     & 618\%   \\
    6     & 875\%   \\
    7     & 1,170\%   \\
    \hline
    \end{tabular}%
  \label{tab:err}%
\end{table}%

\begin{figure}[h]
\centering
\includegraphics[scale=1.0]{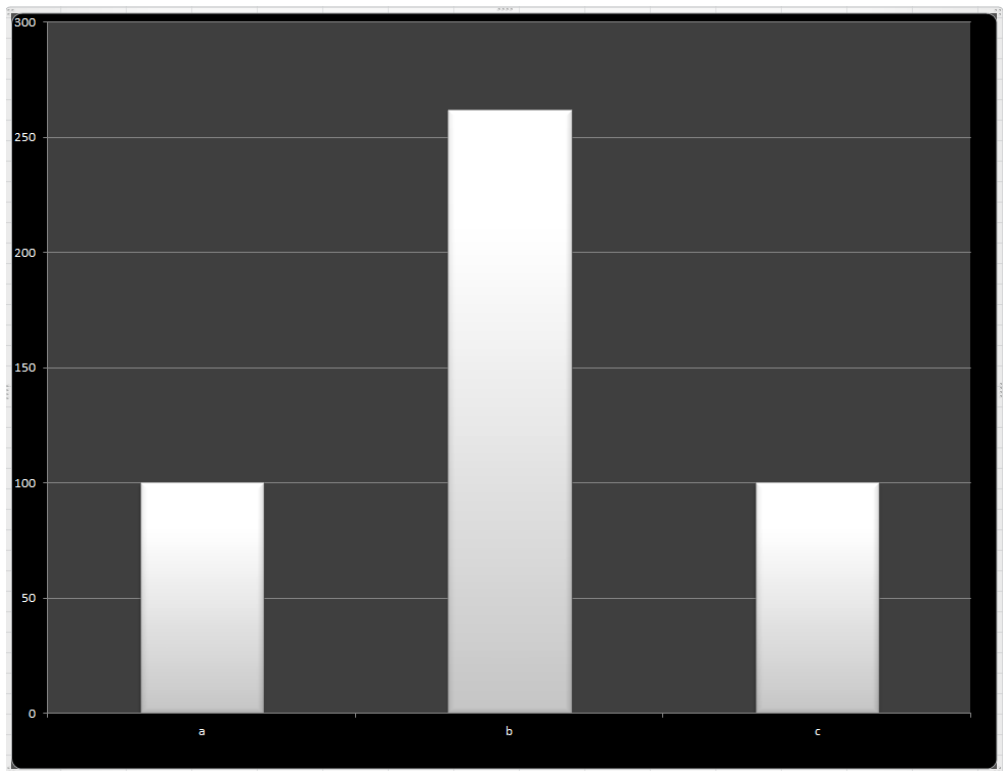}
\caption{Triad $(a,b,c)$ with the 262\% error acceptable by the eigenvalue-based inconsistency for $CPC(2.62,3)$}
\label{fig:bar3}
\end{figure}

\begin{figure}[h]
\centering
\includegraphics[scale=1.0]{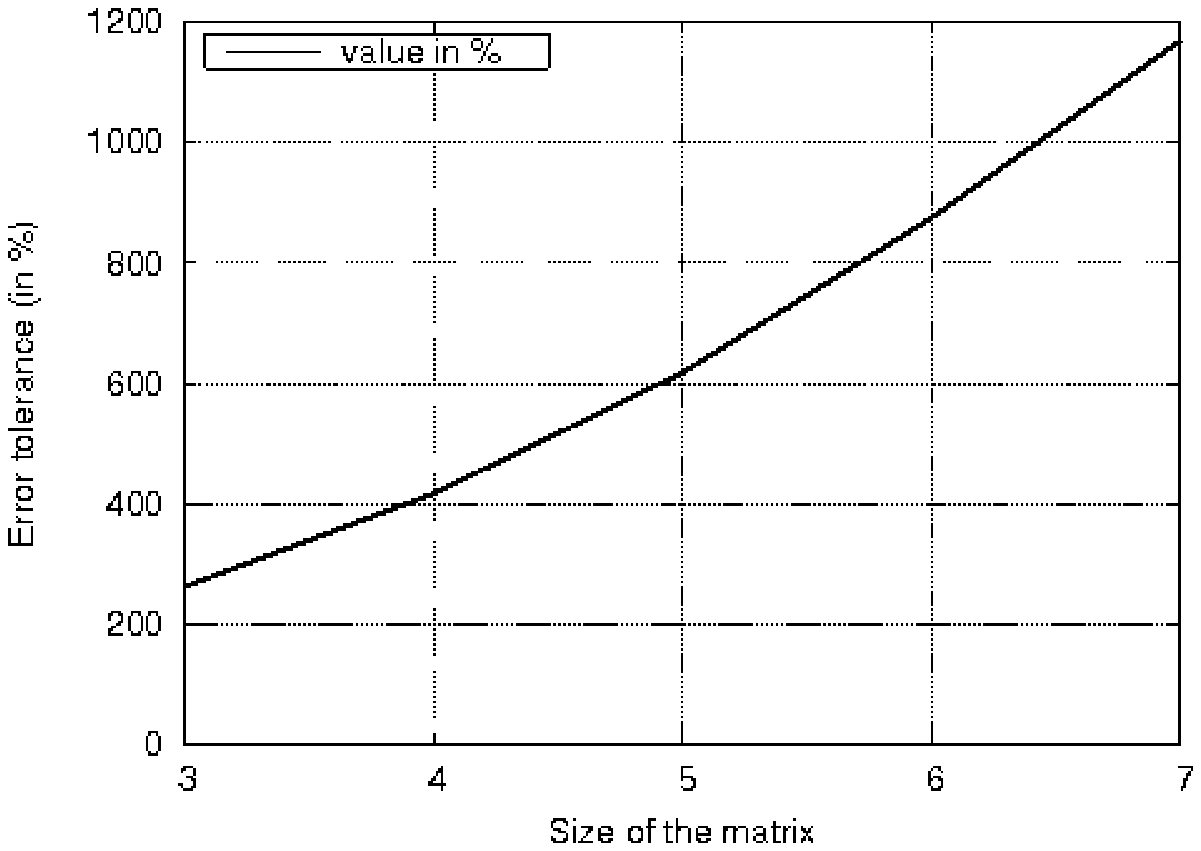}
\caption{Maximal errors acceptable by eigenvalue-based inconsistency for $CPC(x,n)$}
\label{fig:ev-ii}
\end{figure}

$CPC(x,n)$ of the size $n$ by $n$ has $n-2$ triads of this shape: $(1,x,1)$. All triads are formed from these matrix elements  $(a_{ij},a_{ik},a_{jk})$ based the consistency condition is $a_{ik} = a_{ij}* a_{jk}$. Not only the equality does not hold for $x>1$ but for $a_{ij}= a_{jk}=1$ and $x=a_{ij}* a_{jk}$ the inaccuracy grows with the growing $x$. For $CPC(2.62,3)$, it is illustrated by Fig.~\ref{fig:bar3}. The question is evident: ``Would you consider such three bars are equal?'' and if the answer is not, ``why AHP considers such error as acceptable?''

Values $x$ can be an arbitrarily large value which creates a problem. Assuming that the exact values are set to $a_{ij}= a_{jk} =1$, the value $x$ is computed as $a_{ij} * a_{jk} =1$ hence the  error for $x$ is $x/(1*1)$ hence $x$ or $x*100$\%. For example, for $n=7$, $x=4.25$ the error is 1,170\%. However, $x$ can be 1,000,000\%, or more since in Section~\ref{sec:cpc}, the proof has been provided that there is such $n$ for which $CI \leq 0.1$ hence acceptable. The 10\% threshold, originally set as ``the consistency rule'' in \cite{Saaty77} and later on slightly decreased for larger $n$ but it does not matter for the inequality (3) in Section~\ref{sec:cpc} if it is 10\% or any other fixed value.

According the the results in Section~\ref{sec:cpc}, there is always such $n$ for which the deviation of the principal eigenvalue from $n$ is small enough to consider $CPC(x,n)$ matrix acceptable while the arbitrarily large $x$ has $n-2$ triads with an unacceptably high error $x$. 

The distance-based inconsistency was introduced in \cite{Kocz93} and independently analyzed in \cite{BR2008}. 
Its convergence analysis was published in \cite{KS2010}. 
Evidently, it does not accept big values of $x$ in triads $(1,x,1)$. It specifically  postulates to re-examine input data for $ii>1/3$, hence $x>1.5$ is proclaimed to be suspiciously high and the PC matrix needs to be re-examined.

\section{The analysis of $FPC(x,n)$ matrix}
\label{sec:fpc}

We have feared that some of the AHP supporters may hold to the last hope by believing that ``it is only one value in the $CPC(x,n)$ matrix'' since it has x in one matrix element (in fact, $x^{-1}$ in another corner). However, we have a surprise for them by what we call $FPC$ (the ``full'' pairwise comparisons matrix or the PC matrix full of $x$). Unlike $CPC(x,n)$, it has all erroneous triads. 

Consider the matrix $FPC(x,n),$ with $x>1,$ defined by
\begin{displaymath}
FPC(x,n)=\begin{bmatrix}1&x&\cdots&x&x\\
x^{-1}&1&\cdots&x&x\\
\vdots &\vdots &\ddots&\vdots &\vdots\\
x^{-1}&x^{-1}&\cdots&1&x\\
x^{-1}&x^{-1}&\cdots&x^{-1}&1\\\end{bmatrix}\in M_{n\times n}(\mathbb R)
\end{displaymath}

\noindent Let $w$ be the eigenvector corresponding to the principal eigenvalue $\lambda_{max}$. Thus $$x^{-1}(w_1+\ldots+w_{k-1})+w_k +x(w_{k+1}+\ldots+w_n)=\lambda w_k$$ for $k=1,2,\ldots,n$. \\

\noindent Let us notice that for $k=1$, the first term is missing while for $k=n$, the last term is missing.
By subtracting equations corresponding to $k$ and $k-1$, the following holds:

$$x^{-1}w_{k-1}+w_k-w_{k-1}-xw_k=\lambda w_k-\lambda w_{k-1}$$
which gives
$$w_k=w_{k-1} {x^{-1} -1 + \lambda \over x -1 + \lambda}$$
for $k=2,\ldots,n$. \\

\noindent hence

$$w_k= \left( {x^{-1}-1 +\lambda \over x -1 +\lambda }\right)^{k-1}$$

\noindent for $k=1,2,\ldots,n$. \\

\noindent Substituting it into the first equation results in \\

$$1+x(w_2+w_2^2+\ldots+w_2^{n-1})=\lambda $$ 

\noindent hence

$$1+x{w_2^n -w_2 \over w_2-1}=\lambda$$

\noindent by using 
$$ w_2={x^{-1}-1+\lambda \over x-1+\lambda} $$

\noindent and by transforming the last equation, the following equation is obtained:

$$ \left( {x^{-1}-1+\lambda \over x-1+\lambda}\right) ^n = \dfrac{1}{x^2}$$

\noindent therefore 

$$ \lambda = {x-1 \over x} {x+x^{\frac{2}{n}} \over x^{\frac{2}{n}} -1}$$

\noindent \textbf{Example:}\\

\noindent For $x=2.25$ and $n=4$, we have $\lambda_{max}=\frac{25}{6}$

\noindent Thus
$$ {\lambda_{max}-n \over n-1 }= {\frac{25}{6} -4 \over 3} = \frac{1}{18} \approx 0.055555556$$

\noindent therefore 225\% error is still considered as acceptable by AHP theory for $n=4$. The soundness of entering three inaccurate (by 55.6\%) comparisons into the matrix $FPC(x,n)$ and claiming that such matrix is acceptable is left to the reader for his/her evaluation.
 
\noindent For $x=2.84$ and $n=7$, the error increases to 64.79\%. These errors although a bit less impressive than for $CPC(x,n)$ are still by far too high for the estimation lengths of randomly generated bars as it was demonstrated by a Monte Carlo Study in \cite{K1998} where a 5\% error was reported. The error 284\% is bigger than 262\% illustrated in Fig.\ref{fig:bar3}. This study considers it unacceptable. The question is if it is reasonable to consider three bars in Fig.\ref{fig:bar3} as ``equal enough''. The only similar equality of this kind, which comes to our minds is: ``All animals are equal, but some animals are more equal than others.'' [George Orwell, Animal Farm].

\section{Conclusions}

The presented inconsistency axiomatization is simple, elegant, a considerable step forward and a sound mathematical foundation for the further PC research. 
It finally allows us to define proper inconsistency indicators, regardless of whether or not they are localizing the inconsistency or serve as global indicators of inconsistencies in pairwise comparisons matrices. The distance-based inconsistency definition localizes inconsistency and produces correct results. 

The eigenvalue-based consistency index (CI) fails to increase with the growing size of the PC matrix and it has the growing number of triads with each of them having an unacceptable level of inconsistency. As proven in Section~\ref{sec:cpc}, AHP thresholds (both old and recently modified) are unable to detect large quantities of large inaccuracies existing in $CPC(x,n)$ matrices. There is always $n$, for which these inaccuracies are lost in the matrix, no matter how large they are. The discussed eigenvalue-based inconsistency indicator is not precise enough for the detection of individual triads, which turns to be erroneous but ``averaged'' by the eigenvalue processing. It is anticipated that every statistical inconsistency indicator, including those with roots in the principal eigenvalue, may not be good indicators of the problems existing in pairwise comparisons. Simply, they do not look deep enough into relationships existing in cycles of which triads are the most important minimal cycles (as pointed out in this study, one or two elements cannot create an inconsistency cycle).
Hopefully, proponents of other inconsistency indicators will examine their definition by using the proposed axiomatization. Certainly, getting help from authors of this study is a vital solution. 

During the final stages of editing of our study for publication,  the numerical results strongly supporting our finding were located in \cite{XDX2008} with the following text in the conclusions: 

\begin{quotation}
``In this paper, by simulation analysis, we obtain the following
result: as the matrix size increases, the percent of the
matrices with acceptable consistency ($CR \leq 0.1$), decrease
dramatically, but, on the other hand, there will be more and
more contradictory judgments in these sufficiently consistent
matrices. This paradox shows that it is impossible to
find some proper critical values of CR for different matrix
sizes. Thus we argue that Saaty’s consistency test could be
unreasonable.''
\end{quotation}

\noindent It is not a paradox anymore. In this study, the mathematical proof and reasoning for it have been provided.

\section*{Acknowledgment}

This research has been partially supported by the Provincial Government through the Northern Ontario Heritage Fund Corporation and by the Euro Grant Human Capital. The authors would like to thank  Grant O. Duncan (a part-time graduate student at Laurentian University; BI Leader, Health North Sciences, Sudbury, Ontario) for his help with the editorial improvements.


\end{document}